\documentclass[12pt,draftcls,onecolumn]{IEEEtran}
\usepackage{amsmath,amsfonts,amssymb,amsbsy,bm,paralist,theorem,color}
\usepackage{graphicx,algorithmic,algorithm,relsize}
\usepackage{multicol}
\usepackage{multirow}
\usepackage{subfigure}
\usepackage{stfloats}
\usepackage{cite}
\usepackage[square, comma, sort&compress, numbers]{natbib}%����ѹ��
\newtheorem{lem}{Lemma}

\newtheorem{prop}{Proposition}

\graphicspath{{fig/}}

\definecolor{orange}{RGB}{255,107,0}
\definecolor{green}{RGB}{0,160,20}

\begin{document}

\title{Optimal Task Assignment and Power Allocation for NOMA Mobile-Edge Computing Networks}
\author{Fang Fang, Yanqing Xu,
	Zhiguo Ding,
	Chao Shen, Mugen Peng,\\
	and George K. Karagiannidis
	\thanks{F. Fang and Z. Ding are with School of Electrical and Electronic Engineering, The University of Manchester, M13 9PL, UK (e-mail: fang.fang@manchester.ac.uk, zhiguo.ding@manchester.ac.uk).}
	\thanks{Y. Xu and C. Shen are with the State Key Laboratory of Rail Traffic Control
		and Safety, Beijing Jiaotong University, Beijing 100044, China (e-mail: yanqing\_xu@bjtu.edu.cn, chaoshen@bjtu.edu.cn).}
	\thanks{M. Peng is with the Key Laboratory of Universal Wireless Communications,
		Ministry of Education, Beijing University of Posts and Telecommunications,
		Beijing 100876, China (e-mail: pmg@bupt.edu.cn).}

	\thanks{G. K. Karagiannidis is with the Department of Electrical and Computer Engineering, Aristotle University of Thessaloniki, 54124 Thessaloniki,
		Greece (e-mail: geokarag@auth.gr).}
}\maketitle

%##################################################################
\begin{abstract}
Mobile edge computing (MEC) can enhance the computing capability of mobile devices, and non-orthogonal multiple access (NOMA) can provide high data rates. Combining these two technologies can effectively benefit the network with spectrum and energy efficiency. In this paper, we investigate the task completion time minimization in NOMA multiuser MEC networks, where multiple users can offload their tasks simultaneously via the same frequency band. We adopt the \emph{partial} offloading, in which each user can partition its computation task into offloading computing and locally computing parts. We aim to minimize the maximum task latency among users by optimizing their tasks partition ratios and offloading transmit power. By considering the energy consumption and transmitted power limitation of each user, the formulated problem is quasi-convex. Thus, a bisection search (BSS) iterative algorithm is proposed to obtain the minimum task completion time. To reduce the complexity of the BSS algorithm and evaluate its optimality, we further derive the closed-form expressions of the optimal task partition ratio and offloading power for two-user NOMA MEC networks based on the analysed results. Simulation results demonstrate the convergence and optimality of the proposed a BSS algorithm and the effectiveness of the proposed optimal derivation. %It can also be concluded that the NOMA multiuser MEC can realize lower task completion time compared OMA to MEC networks. 
%high energy efficiency provided by the proposed power allocation scheme.
\end{abstract}

\section{Introduction}
Driven by the explosive emergence of new compute-intensive applications in the Internet of things (IoT), especially ultra-low-latency applications, such as virtual reality (VR) and augmented reality (AR), mobile edge computing (MEC) was proposed to enhance the computing capability of the mobile devices \cite{PMMECSurvey2017}. In order to meet extremely high data rate requirements, non-orthogonal multiple access (NOMA) has been considered as another promising technology in the next generation of wireless communications, due to its superior spectrum efficiency \cite{DingIEEEMag16,DingJSACSurvey2017,DaiSurvey2018}. Specifically, the multiuser superposition transmission scheme, which is a special case of NOMA, has been adopted in the third generation partnership project long term evolution advanced (3GPP-LTE-A) networks \cite{3GPP}. Therefore, to support multiple users and lower the transmission latency and energy consumption, the combination of NOMA and MEC is an inevitable trend in the next generation of wireless networks \cite{AKiani2018JIOT}.

\subsection{Related literature}
 In MEC communication systems, base stations (BSs) equipped with MEC servers, can provide cloud-like computing for the mobile users, which have computation intensive and delay sensitive tasks \cite{WShiIEEEIoT2016}. Due to the enhancement of computational capacity, MEC has been regarded as a key technology  for the next generation of  networks \cite{YMaoMECSureveys2017}. In multiuser MEC networks, due to the low computational capacity of the mobile devices, the majority of computing tasks can be offloaded to MEC server based BSs for remote computation. A MEC server computes the tasks much faster than the mobile devices, due to its considerable computing capability. After task computation, tasks results can be downloaded from the MEC server to the mobile users through downlinks \cite{SAMECSurvey2014,HZhangFogIEEEWC2017}. In the open literature, computational resource allocation (e.g., tasks assignment) and communication resource allocation (e.g., offloading power and subchannel allocation) have been investigated to improve the performance of multiuser MEC networks. In MEC offloading transmission, the tasks can be \emph{binary} offloaded (i.e., the computation task cannot be partitioned and must be either fully offloaded to the MEC server or computed locally) \cite{YMaoTWC2017,TXTranTVT2018,FWangTWC2018,SBIEEETWC2018,QPhamAccess2018,QPhamAccess19} and the \emph{partial} offloaded (i.e., the computation task of each user can be partitioned into two parts, one of which will be offloaded for remote computation, and the other part can be computed locally by the mobile devices) \cite{CYouTWC2017}. 

A very popular scheme is the power-domain NOMA, where the channel difference between users is exploited to multiplex different users on the same frequency band with different power levels \cite{SaitoIEEENoma13,DingIEEEMag16}. Thus multiple users can transmit signals simultaneously with lower interference than the orthogonal multiple access (OMA) system. This attracts lots of researchers' attention in recent years. Resource optimization, i.e. subchannel allocation, power allocation and user association, has been investigated to enhance the performance of NOMA networks \cite{YSunIEEETCOM17,FangIEEETrans16,FangJSAC17,HZhangMagzine2018}. Among these works, the weighted sum rate was maximized for the downlink full-duplex NOMA system via the proposed optimal power allocation and subchannel allocation in \cite{YSunIEEETCOM17}. To meet the requirement of green communications, the energy efficient resource allocation was investigated for the perfect channel statement (CSI) \cite{FangIEEETrans16} and the imperfect CSI \cite{FangJSAC17,HZhangMagzine2018}.

Even though  MEC and NOMA have been widely studied in recent years, only handful research works have been carried to study the combination of NOMA and multiuser MEC networks  \cite{ZDing2018TCOM,FWang2018TCOM,ZDing2018WCL,ZDing2018WCLDely,AKiani2018JIOT,YPanCL2018,XDiaoAccess2019D2D,LPSIC2018IoT}. To further enhance the spectrum efficiency and the performance of multiuser MEC networks, NOMA uplink transmission and downlink transmission were proposed to be applied to MEC \cite{ZDing2018TCOM}. In this system, multiple users can offload (download) their tasks (task results) to (from) the MEC server simultaneously via the same frequency band. The performance comparison of NOMA-MEC and OMA-MEC demonstrates that the NOMA-MEC system can achieve superior performance on latency reduction and energy consumption reduction \cite{ZDing2018TCOM}. Besides, resource allocation research works have been carried out to enhance the NOMA enable MEC system performance \cite{ZDing2018WCL,ZDing2018WCLDely,AKiani2018JIOT}. Among these research works, the energy consumption minimization was investigated for a multiuser NOMA-MEC system \cite{FWang2018TCOM,XDiaoAccess2019D2D}. In \cite{FWang2018TCOM}, an efficient algorithm of optimizing offloading tasks and offloading power levels of each user was proposed to minimize the system energy consumption for NOMA based MEC system \cite{FWang2018TCOM} and D2D-assisted and NOMA-based MEC system \cite{XDiaoAccess2019D2D}, respectively. Different from OMA-MEC and pure NOMA-MEC systems proposed in \cite{FWang2018TCOM,ZDing2018TCOM}, a hybrid  NOMA  scheme was proposed in \cite{ZDing2018WCL}, where a user can first offload parts of its task within the time slot allocated to another user and then offload the remaining task during the time slot solely occupied by itself. In this work, the power allocation and time allocation were optimized to minimize the energy consumption of offloading for a NOMA multiuser MEC network. Subsequently, to minimize the task delay, Dinkelbach's method and Newton's method were compared for the hybrid NOMA-MEC system \cite{ZDing2018WCLDely}, in which the simulation results illustrate that both algorithms converge to the same optimal point, but Netwon's method converges faster than Dinkelbach's method. Unlike the \emph{partial} offloading scheme, the authors in \cite{AKiani2018JIOT,FangGlobecome2019} focused on the independent and non-separate task. In this work, an efficient heuristic algorithm of user clustering and frequency and resource blocks allocation was proposed to minimize the system energy consumption \cite{AKiani2018JIOT}.
The energy efficient power allocation, time allocation and task assignment were proposed to minimize the energy consumption for a NOMA MEC network \cite{YPanCL2018}. Besides the computational resource, the successive interference cancellation (SIC) decoding order was optimized to reduce the task delay for NOMA enabled narrowband Internet of Things (NB-IoT) systems \cite{LPSIC2018IoT}.

\subsection{Contributions}
Different from the existing works, which mainly focus on energy minimization for NOMA multiuser MEC networks, we focus on the completion time minimization by considering the \emph{partial} offloading in NOMA multiuser MEC. The main contributions are listed as follows:
\begin{itemize}
\item In this paper, we apply NOMA into a multiuser MEC network where multiple users can offloading their task to the MEC server simultaneously via the same frequency band. Considering the \emph{partial} task offloading scheme, the energy consumption limitation and offloading power limitation, the task completion time minimization problem is formulated as a non-convex problem. Thus it is polynomial time unsolvable. By analysing the properties of the formulated problem, some significant insights are revealed, and the corresponding propositions and Lemma are proposed to equivalently transform the original formulated problem to a simplified form. Based on those analytical results, the quasi-convexity of the transformed problem is proved. Therefore, a bisection searching  (BSS) based algorithm is proposed to find the global optimal solution to the transformed problem. In the proposed algorithm, the original problem is solved by equivalently solving a series of feasibility subproblems. Moreover, we analyze the complexity of the proposed algorithm.  The convergence and optimality of BSS algorithm are evaluated by simulation results.  
\item Motivated by the practical applications, we focus on the two-user case to reduce the decoding complexity of SIC. To further reduce the complexity of the proposed BSS iterative algorithm, a corresponding proposition is proposed to equivalently transform the original problem to a convex problem. The convexity of the transformed problem is proved. Moreover, the closed-form expressions of the task partition ratios and offloading power are derived by exploiting the Lagrangian approach, i.e., Karush-Kuhn-Tucker (KKT) conditions, for the two-user NOMA MEC network. The simulation results demonstrate the proposed BSS iterative algorithm matches with our derived optimal solution, which reveals that our proposed BSS algorithm can converge to the optimal solution. %Compared with existing schemes, i.e., fully NOMA system and fully OMA system, our proposed algorithm can provide the minimum task completion time in the simulation results. 
\end{itemize}

\subsection{Organization}

The organization of this paper is as follows. We introduce
the system model and formulate the task completion time minimization problem in Section II. The BSS algorithm is proposed in Section III. In Section IV, the optimal closed-form solution derived for two-user case. Simulation results are presented in Section IV, and Section V concludes this paper.

\section{System Model and Problem Formulation}
\subsection{NOMA-enabled Multiuser MEC Networks}
In the NOMA multiuser MEC network, $M$ users are randomly distributed in the single cell where one BS equipped with the MEC server is located in the cell centre. We assume that all users and the BS are equipped with single antenna. The indices of users are defined as $m\in\{1,2,\cdots,M\}$. Denote the offloading task partial factor of User $m$ ($U_m$) by $\beta_m$, thus $(1-\beta_m)$ is the task partial factor of the task to locally compute at $U_m$, where $\beta_m \in (0,1)$. By implementing NOMA, the SIC technique is applied at the MEC based BS. Define ${h}_{m}$ as the channel gain from $U_m$ to the MEC server. Without loss of generality, $M$ users are sorted as $|h_{1}|\leq|h_{2}|\leq \cdots \leq |h_{M}|$. In NOMA offloading transmission, the SIC decoding order is assumed as the decreasing order of channel gains. It indicates that the MEC server first decodes the information transmitted by $U_M$ and then decodes the information of $U_{M-1}$, until $U_1$. Define $p_{m}$ as the offloading power of $U_m$, then the signal-to-interference-plus-noise-ratio (SINR) of the $U_m$ received at the MEC server can be written by
\begin{equation}
\Gamma_{m}^{off}=\frac{|h_{m}|^2p_{m}}{\sum \limits_{j=1}^{m-1}|h_{j}|^2p_{j}+\sigma^2}.
\end{equation}
where $\sigma^2$ represents zero-mean complex additive white Gaussian noise (AWGN) power. Denote the bandwidth of this system by $B$ Hz, thus the achievable data rate of the $U_m$ can be written by
\begin{equation}\label{R_m_n_off}
\begin{aligned}
R_{m}=&B\log_2(1+\Gamma_{m}^{off})
=B\log_2\left(\frac{\sum \limits_{i=1}^{m}|h_{i}|^2p_{i}+\sigma^2}{\sum \limits_{j=1}^{m-1}|h_{j}|^2p_{j}+\sigma^2}\right).
\end{aligned}
\end{equation}

In the MEC computation model, each user is enabled to offload a part of its task to the MEC server, where the offloaded task can be computed, and then the result can be downloaded from the MEC server to users. The task of $U_m$ can be described by two parameters $(L_m, C_m)$ where $L_m$ is the input number of bits for this task, and $C_m$ denotes the number of CPU cycles required to compute one bit of this task. In this paper, we assume that the MEC server based BS has sufficient computational capability and resource. Therefore, the remote execution time at MEC and downloading time can be negligible \cite{ZDing2018WCL,ZDing2018WCLDely}. 

\subsubsection{Time Consumption for Offloading }
In this phase, each user will offload part of its task to the MEC server for remote executions. According the achievable offloading data rate \eqref{R_m_n_off} of the $U_m$, the task offloading time from $U_m$ to the MEC server can be written by 
 \begin{equation}\label{T_m_off}
  \begin{aligned}
 T_{m}^{off}=\frac{\beta_{m}L_m}{R_{m}}.
  \end{aligned}
  \end{equation}
The offloading energy consumption for $U_m$ is
 \begin{equation}\label{E_m_off}
\begin{aligned}
E_m^{off}=T_{m}^{off}p_{m}.
\end{aligned}
\end{equation}
\subsubsection{Mobile Execution Time}
For $U_m$, partial task $\beta_mL_m$ is offloaded to the MEC server for remote computation, and the remaining task $(1-\beta_m)L_m$ is computed locally. Denote the CPU frequency at $U_m$ is $f_m^{loc}$ (in cycles per second). The local computation time of $U_m$ is given by
  \begin{equation}\label{T_m_loc}
 \begin{aligned}
 T_{m}^{loc}=\frac{(1-\beta_{m})L_mC_m}{f_m^{loc}}.
 \end{aligned}
 \end{equation}
Thus the energy consumption of computing at the mobile device $U_m$ can be written by
\begin{equation}\label{E_m_comp}
\begin{aligned}
E_{m}^{c}=\kappa_m(1- \beta_m) L_mC_m (f_m^{loc})^2
\end{aligned}
\end{equation}
where $\kappa_m$ denotes the effective capacitance coefficient for each CPU cycle of the local user $U_m$.

\subsection{Problem Formulation}
For each user, partial task will be offloaded to the MEC server for remote computation, and the remaining task will be computed locally. Each user's task will cost time to execute. Therefore, in this paper, we aim to minimize the maximum task completion time among users by optimizing task assignment (i.e., offloading task ratio $\beta_m$) and communication resource (i.e., the offloading transmit power $p_m$). The computing time at MEC server can be ignored compared to the offloading time due to the high computing capacity of MEC server. 
By considering the local computing time, the task completion time of $U_m$ can be written by
\begin{equation}\label{T_m}
T_{m}= \max \left\{T_{m}^{off},\ T_m^{loc}\right\}.
\end{equation}
Therefore, the task completion time minimization problem can mathematically formulated as
\begin{subequations}\label{Prob:T_min}
\begin{align}
 	\mathop {\min\ \max }\limits_{{\{ \bm{\beta},\bm{p}\}}}  &\left\{T_{m}^{off},T_m^{loc},\forall m\right\} \\
 	 \text{s.t.}  \quad \quad \quad &0 \leq \beta_{m}\leq 1, \forall m,\label{eq:beta_con}\\   
    &0\leq p_{m} \leq P_{\max}, \forall m,\label{eq:p_con}\\ 
 	&E_m^c+E_m^{off}\leq E_{\max}, \forall m,\label{eq:E_m_con} 
% 	&E_{MEC}^c+E_{MEC}^{down}\leq E_{\max}  \label{eq:E_MEC_con}
\end{align}
\end{subequations}
where $\bm{\beta}=[\beta_1,\cdots,\beta_M]^T$ and $\bm{p}=[p_1,\cdots,p_M]^T$. Constraint \eqref{eq:beta_con} specifies the range of the offloading task ratio.  Constraint \eqref{eq:p_con} describes the range of the offloading transmit power. Constraint \eqref{eq:E_m_con} guarantees that the energy consumed at each mobile user is limited to the maximum energy consumption $E_{\max}$. Since problem \eqref{Prob:T_min} is nonconvex, we propose the following transformation and solution to solve this problem.

  \section{The Optimal Solution for the Multiuser case}
 \subsection{Significant Observations}
  In this section, we consider the multiuser case and attain the minimum computation time for NOMA multiuser MEC networks. An bisection searching based iterative algorithm is proposed to find the global solution to problem \eqref{Prob:T_min}. Before solving this problem, we propose the following Lemma and proposition to simplify the optimization problem \eqref{Prob:T_min}.
 
 \begin{lem}\label{Lemma1}
 	In NOMA multiuser MEC networks, if $M$ users offload their signals to the BS equipped with MEC server within the same transmit time $T$,
 	\begin{equation}\label{Tm}
     T=T_{m}^{off}=T_{m'}^{off},\ \forall m\neq m'. 
 	\end{equation}
 	Then \eqref{Tm} can be equivalently transformed to 
 	\begin{equation}
 	\tilde{T}_m^{off}=\frac{\sum\limits_{i=1}^{m}\beta_{i}L_i}{B\log_2\left( \sum\limits_{i=1}^m|h_{i}|^2 p_{i} +\sigma^2\right)}, \quad \forall m.  
 	\end{equation}
 \end{lem}
 \begin{proof}
 	The proof is shown in Appendix \ref{Lemma1Proof}.
 \end{proof}
 According to Lemma \ref{Lemma1}, we propose the following proposition to further simplify problem \eqref{Prob:T_min}.
 \begin{prop} \label{T1=T2}
 	For any two users ($U_m$ and $U_{m'}$ $\forall m\neq m'$) in NOMA multiuser MEC networks, to minimize the maximum task completion time of different users, i.e.,
 	\begin{equation}\label{T_m^off}
 	\mathop {\min\ \max }\limits_{\{\beta_m,\beta_{m'},p_m,p_{m'}\}}\left\{T_{m}^{off},\ T_{m'}^{off}\right\},
 	\end{equation}
 	the optimal solution $\{\beta_m^*,\beta_{m'}^*,p_m^*,p_{m'}^*\}$ can be only obtained if and only if the offloading time of each user equals to each other $T_m^{off}=T_{m'}^{off}$.
 \end{prop}
 \begin{proof}
 	The proof is shown in Appendix \ref{T_1_off=T_m_off}.
 \end{proof}
\subsection{BSS Iterative Algorithm}
 Based on Lamma \ref{Lemma1} and Proposition \ref{T1=T2}, we have $\tilde{T}_m^{off}=T_1^{off}=T_2^{off}$. The energy consumed at $U_m$ can be written as
 \begin{equation}
 E_m=  \kappa_m \beta_m L_m (f_m^{loc})^2+\tilde{T}_m^{off}p_m, \forall m.
 \end{equation}
Thus, the task completion minimization problem for NOMA multiuser MEC can be reformulated as:
\begin{subequations}\label{Prob:T_min_2}
	\begin{align}
	\mathop {\min \quad \max }\limits_{ \{\bm{\beta},\bm{p}\} }  &\left\{\tilde{T}_m^{off},T_m^{loc},\forall m\right\} \\
	\text{s.t.}  \quad \quad \quad &\beta_{m}\in [0,1], \forall m, \label{eq:beta_con_2}\\   
	&0\leq p_{m}\leq P_{\max},\forall m,\label{eq:p_con_2}\\ 
	&E_m\leq E_{\max},\forall m.\label{eq:E_1_con}
%	&E_2^c+E_2^{off}\leq E_{\max}, \label{eq:E_2_con}  
	\end{align}
\end{subequations}
 This problem is  non-convex due to the nonconvexity of $\tilde{T}_m^{off}$ with respect to $\{\beta_m,p_m\}$. Thus it is challenging to obtain its global optimum within polynomial time. However, after investigating properties of this problem, we first propose Proposition \ref{quasi} to demonstrate the quasi-convexity of the objective function in \eqref{Prob:T_min_2}, which can be solved by a series of convex feasibility problems. 
\begin{prop}\label{quasi}
	Given by the expressions \eqref{T_m_loc} and \eqref{T_m^off}, the objective function in problem \eqref{Prob:T_min_2} is strictly quasi-convex. 
\end{prop}
\begin{proof}
	The proof is shown in Appendix \ref{QuasiProof}.
\end{proof} 
To obtain the optimal solution, problem \eqref{Prob:T_min_2} can be equivalently transformed to the following problem:
\begin{subequations}\label{Prob:T_transfer}
	\begin{align}
	\mathop {\min}\limits_{\{\bm{\beta},\bm{p},\alpha_T \}} \quad& \alpha_T \\
\text{s.t.}  \quad &\frac{\sum\limits_{i=1}^m\beta_{i} L_{i}}{\log_2\left(1+ \sum\limits_{i=1}^M|h_{i}|^2 p_{i} \right)}\leq \alpha_T,\
\forall m \label{R_alphaT_M}\\
&\frac{(1-\beta_m)L_mC_m}{f_m^{loc}}\leq \alpha_T, \
\forall m \label{bm_alphaT} \\
& 0\leq\beta_{m}\leq 1,\  \forall m \label{b_0_1_M}\\
& 0\leq p_m\leq P_{\max},\  \forall m \label{p_0_max_M}\\
& \kappa_m(1- \beta_m) L_mC_m (f_m^{loc})^2+\alpha_Tp_m\leq E_{\max}, \ \forall m.\label{EE_m}\\\nonumber
	\end{align}
\end{subequations}
Note that the inequality constraint set is not convex set in $\{\bm{\beta} ,\bm{p} \}$ due to the quasi-convexity of \eqref{R_alphaT_M} and non-convexity of \eqref{EE_m}. However, problem \eqref{Prob:T_transfer} becomes a feasibility problem if we fix $\alpha_T$. Thus we can solve the above problem by solving a serious of convex feasibility subproblems. For a given $\alpha_T$, the feasibility problem can be formulated as:
\begin{subequations}\label{Prob:T_feasible}
	\begin{align}
	\text{ {find }}\quad&\{\bm{\beta} ,\bm{p} \} \\
	\text{s.t.} \quad&\eqref{R_alphaT_M}-\eqref{EE_m}.
	\end{align}
\end{subequations}
Note that the convex constraint set can be denoted as 
\begin{equation}
	C_{\alpha_T}=\{\{\bm{\beta} ,\bm{p} \}|\eqref{R_alphaT_M}-\eqref{EE_m}\}, \forall \alpha_T.
\end{equation}
Define the optimal solution of this feasibility problem is $\alpha_T^*$. Bisection search method can be utilized to find $\alpha_T^*$ \cite{YXuTSP2017}. Therefore, Algorithm \ref{Al1} is proposed to find the minimum task completion time of user tasks. In Algorithm \ref{Al1}, we first initialize $\alpha_T$ by its lower bound and upper bound. For given $\alpha_T$, we say that problem \eqref{Prob:T_feasible} is feasible $C_{\alpha_T}\neq\phi$, and we have $\alpha_T\geq \alpha_T^*$. Problem \eqref{Prob:T_feasible} is infeasible $C_{\alpha_T}=\phi$, and we have  $\alpha_T<\alpha_T^*$. Algorithm \ref{Al1} converges to the unique global optimal solution to problem due to its strictly quasi-convexity \cite{ChongChiConvex}.

	\begin{algorithm}[!t] \small
	\caption{~BSS-based algorithm for problem \eqref{Prob:T_transfer} }\label{Alg1}
	\begin{algorithmic}[1] 
		\STATE {{\bf Initialization:} Set $\alpha_T^{\min}=0$, $\alpha_T^{\max}=\max\{\frac{L_1C_1}{f_1^{loc}},\cdots,\frac{L_MC_M}{f_M^{loc}}\}$ and the accuracy $\epsilon=10^{-4}$.}
		%\STATE {\bf repeat}
		\WHILE {$\alpha_T^{\max}-\alpha_T^{\min}> \epsilon$}
		\STATE {Set $\alpha_T=(\alpha_T^{\min}+\alpha_T^{\max})/2$.  }
		\STATE {Solve the convex feasibility problem \eqref{Prob:T_feasible} and find $C_{\alpha_T}$. }
		\IF{$C_{\alpha_T}\neq \phi$}
		\STATE Update $\alpha_T^{\max}=\alpha_T$.
		\ELSE
		\STATE Update $\alpha_T^{\min}=\alpha_T$.
		\ENDIF
		\ENDWHILE
		\STATE {{\bf Output:} $\alpha_T^*=\frac{\alpha_T^{\min}+\alpha_T^{\max}}{2}$, $\bm{\beta}^*$ and $\bm{p}^*$.}
	\end{algorithmic}\label{Al1}
\end{algorithm}

The computational complexity of Alg. \ref{Alg1} mainly comes from the BSS to find the optimal $\alpha_T$. For given accuracy $\epsilon$, $\alpha_T^{\min}$ and $\alpha_T^{\max}$, the computational complexity of Alg. \ref{Alg1} is given by $\mathcal{O}\left(\log_2\left(\frac{\alpha_T^{\max} - \alpha_T^{\min}}{\epsilon}\right)\right)$.

\section{Closed-Form Optimal Solution Derivation for the Two-user case} 
To further reduce the complexity of the proposed iterative algorithm, in this section, we derive the closed-form solution for two-user case based on the insights and propositions obtained from problem  \eqref{Prob:T_min_2}. To further simplify the problem, the following proposition can be obtained.
 
\begin{prop} \label{T_loc=T_off}
	For each user $U_m$ with its offloading task ratio $\beta_m$,  to minimize the maximum of its offloading time $T_m^{off}$ and local computing time $T_m^{loc}$, i.e.,
	\begin{equation}\label{T_m^loc}
	\mathop {\min\ \max }\limits_{\{\beta_m,p_m\}}\left\{T_{m}^{off},\ T_m^{loc}\right\}\  \forall m,
	\end{equation}
	the optimal solution $\{\beta_m^*,p_m^*\}$ can be only obtained if and only if its offloading time equals to its local computing time, i.e., $T_m^{off}=T_m^{loc}$. 
\end{prop}
\begin{proof}
	Proof by contradiction is shown in Appendix \ref{T_m_loc=T_m_off}. 
\end{proof}
Based on Proposition \ref{T_loc=T_off}, we conclude that the optimal solution to problem \eqref{Prob:T_min_2} can be obtained when $T_m^{off}=T_m^{loc}$ for each user. According to Proposition \ref{T1=T2}, the optimal solution can only be obtained when the offloading time equals to each other. Considering two-user case $|h_1|\leq|h_2|$, problem \eqref{Prob:T_min_2} can be rewritten by
\begin{subequations}\label{Prob:T_min_2_2}
	\begin{align}
	\mathop {\min }\limits_{ \{\beta _{1},\beta _{2},p_{1}, p_{2} \}}  &\frac{\beta_1L_1+\beta_2L_2}{\log_2(1+|h_1|^2p_1+|h_2|^2p_2)} \\
	\text{s.t.}  \quad \quad \quad &\beta_{1}\in [0,1], \beta_{2}\in [0,1], \\   
	&0\leq p_{1}\leq P_{\max},0\leq p_{2}\leq P_{\max}, \label{p_range}\\ \nonumber
	&\kappa_1(1- \beta_1) L_1C_1 (f_1^{loc})^2\\
    &+\frac{\beta_1L_1+\beta_2L_2}{\log_2(1+|h_1|^2p_1+|h_2|^2p_2)}p_1 \leq E_{\max}, \label{eq:E1}\\\nonumber
    &\kappa_2(1- \beta_2) L_2C_2 (f_2^{loc})^2 \\
	&+\frac{\beta_1L_1+\beta_2L_2}{\log_2(1+|h_1|^2p_1+|h_2|^2p_2)}p_2 \leq E_{\max}\label{eq:E2}\\ 
	&\frac{(1-\beta_1)L_1C_1}{f_1^{loc}}=\frac{\beta_1L_1+\beta_2L_2}{\log_2(1+|h_1|^2p_1+|h_2|^2p_2)} \label{eq1}\\
	&\frac{(1-\beta_2)L_2C_2}{f_2^{loc}}=\frac{\beta_1L_1+\beta_2L_2}{\log_2(1+|h_1|^2p_1+|h_2|^2p_2)}\label{eq2}\\
    &\frac{\beta_1L_1}{\log_2(1+|h_1|^2p_1)}=\frac{\beta_1L_1+\beta_2L_2}{\log_2(1+|h_1|^2p_1+|h_2|^2p_2)}\label{eq3}.\\\nonumber
	\end{align}
\end{subequations}
In problem \eqref{Prob:T_min_2_2}, the objective function in \eqref{Prob:T_min_2_2} is quasiconvex based on Proposition \ref{quasi}. However, the constraint \eqref{eq:E1} and \eqref{eq:E2} are not convex set with respective to $\{\beta _{1},\beta _{2},p_{1}, p_{2}\}$. To simplify this problem, we first deal with equality constraints \eqref{eq1}-\eqref{eq3}. To solve the above problem and obtain the global optimum, we first equally transform this problem to an equivalent convex form via equality constraints. By using the equation \eqref{eq3}, we can replace the right sides of \eqref{eq1} and \eqref{eq2} with the left side of \eqref{eq3}. Then we have
 \begin{subequations}\label{equalities}
 		\begin{align}
	&(1-\beta_1)L_1C_1=\frac{\beta_1L_1}{\log_2(1+|h_1|^2p_1)}f_1^{loc} \label{eq:beta11}\\
&(1-\beta_1)L_1C_1/f_1^{loc}=(1-\beta_2)L_2C_2/f_2^{loc}\\
&(1-\beta_1)L_1C_1/f_1^{loc}=\frac{\beta_1L_1+\beta_2L_2}{\log_2(1+|h_1|^2p_1+|h_2|^2p_2)}. 
	\end{align}
\end{subequations}
After a series of calculations, the objective function in \eqref{Prob:T_min_2_2} can be rewritten by
\begin{equation}
\frac{L_1+L_2}{\frac{f_1^{loc}}{C_1}+\frac{f_2^{loc}}{C_2}+B\log_2(1+|h_1|^2p_1+|h_2|^2p_2)}.
\end{equation} 
The transformation can be found in Appendix \ref{ObjTransfer}.
Therefore, problem \eqref{Prob:T_min_2_2} can be rewritten by 
\begin{subequations}\label{Prob:p1,p2}
	\begin{align}
	\mathop {\min }\limits_{ \{p_{1},p_{2}\}} \quad &\frac{L_1+L_2}{\frac{f_1^{loc}}{C_1}+\frac{f_2^{loc}}{C_2}+B\log_2(1+|h_1|^2p_1+|h_2|^2p_2)}\\
	\text{s.t.}  \quad&\eqref{p_range}-\eqref{eq:E2}.    
	\end{align}
\end{subequations}
\begin{prop} Problem \eqref{Prob:p1,p2} is convex problem.
\end{prop}
\begin{proof}
 The convexity proof is shown in Appendix \ref{ConvexityProof}.
\end{proof}
 According to equations \eqref{eq1}-\eqref{eq3}, once the optimal $p_1^*$ and $p_2^*$ are obtained, the optimal $\beta_1^*$ and $\beta_2^*$ can be calculated by the following expressions.
\begin{subequations}
	\begin{align}
	&\beta_1^*=\frac{\log_2(1+|h_1|^2p_1^*)}{\frac{f_1^{loc}}{C_1}+\log_2(1+|h_1|^2p_1^*)}\\
	&\beta_2^*=1-\frac{(1-\beta_1^*)L_1C_1f_2^{loc}}{L_2C_2f_1^{loc}}\\ \nonumber
	\end{align}
\end{subequations}

In the following, we focus on deriving the optimal closed-form expressions of $p_1^*$ and $p_2^*$. Since problem \eqref{Prob:p1,p2} is convex and satisfy Slater's condition, the KKT conditions can be exploited to derive the closed-form optimal solution. The optimal solution can be concluded obtained by following four cases.

\emph{\underline{Case 1}}: When \begin{equation}
\left\{
\begin{aligned}
&  P_{1,w}(p_2^*=P_{\max}) \geq P_{\max} \quad  & \\
&  P_{2,w}(p_1^*=P_{\max}) \geq P_{\max} &
\end{aligned}
\right.
\end{equation}
Define
\begin{equation}
P_{1,w}(p_2^*)= -\frac{\mathcal{W}_0\left(-\frac{B_1\log(2)}{|h_1|^2}2^{\left(-\frac{B_1}{|h_1|^2}+A_1\right)}\right)}{B_1\log(2)}-\frac{1+|h_2|^2p_2^*}{|h_1|^2},\label{p1range_p1}
\end{equation} 
where
$A_1=\frac{\kappa_1a_1(f_1^{loc})^3}{E_{\max}B}-\frac{b_1}{B}$, $B_1=\frac{a_1}{E_{\max}B}$, $a_1=L_1+L_2$ , $b_1=\frac{f_1^{loc}}{C_1}+\frac{f_2^{loc}}{C_2}$ and where $\mathcal{W}_0(\cdot)$ is Lambert $\mathcal{W}$ function, which is single value function.

  \begin{equation}
 P_{2,w}(p_1^*)= -\frac{\mathcal{W}_0\left(-\frac{B_2\log(2)}{|h_2|^2}2^{\left(-\frac{B_2}{|h_2|^2}+A_2\right)}\right)}{B_2\log(2)}-\frac{1+|h_1|^2p_1^*}{|h_2|^2}\label{p2range_p1}
 \end{equation} 
 where $A_2=\frac{\kappa_2a_1(f_2^{loc})^3}{E_{\max}B}-\frac{b_1}{B}$ and $B_2=\frac{a_1}{E_{\max}B}$. 
Thus we have
\begin{equation}\label{Case1}
\left\{
\begin{aligned}
&p_1^*= P_{\max} \quad  &\\
&p_2^*=P_{\max}. \quad   &
\end{aligned}
\right.
\end{equation}

%$P_{1,w}(p_2^*)$ and $P_{2,w}(p_1^*)$ are defined in \eqref{p1range_p1} and \eqref{p1range_p1}.
\emph{\underline{Case 2}}: When \begin{equation}
\left\{
\begin{aligned}
&  P_{1,w}(p_2^*=P_{\max}) \geq P_{\max}\quad  & \\
& P_{2,w}(p_1^*=P_{\max}) \leq  P_{\max}, &
\end{aligned}
\right.
\end{equation}
thus we have
\begin{equation}\label{Case2}
\left\{
\begin{aligned}
&p_1^*= P_{\max} \quad  &\\
&p_2^*=P_{2,w}(p_1^*=P_{\max}). \quad   &
\end{aligned}
\right.
\end{equation}

\emph{\underline{Case 3}}: When \begin{equation}
\left\{
\begin{aligned}
& P_{1,w}(p_2^*=P_{\max})\leq P_{\max} \quad  & \\
& P_{2,w}(p_1^*=P_{1,w}) \geq P_{\max}, &
\end{aligned}
\right.
\end{equation}
thus we have
\begin{equation}\label{Case3}
\left\{
\begin{aligned}
&p_1^*= P_{\max} \quad  &\\
&p_2^*=P_{2,w}(p_1^*=P_{\max}). \quad   &
\end{aligned}
\right.
\end{equation}

\emph{\underline{Case 4}}: When \begin{equation}
\left\{
\begin{aligned}
& P_{1,w}(p_2^*)\leq P_{\max} \quad  & \\
& P_{2,w}(p_1^*)\leq P_{\max}, &
\end{aligned}
\right.
\end{equation}
thus we have
\begin{equation}\label{Case44}
\left\{
\begin{aligned}
p_1^*=&\kappa_2(f_2^{loc})^3-\kappa_1(f_1^{loc})^3+p_2^*\\
p_2^*=&-\frac{\mathcal{W}_0\left(-\frac{B_2\log(2)}{(|h_1|^2+|h_2|^2)}2^{\left(-\frac{B_2}{(|h_1|^2+|h_2|^2)}+A_2\right)}\right)}{B_2\log(2)}\\
&-\frac{1+|h_1|^2\left(\kappa_2(f_2^{loc})^3-\kappa_1(f_1^{loc})^3\right)}{(|h_1|^2+|h_2|^2)}. 
\end{aligned}
\right.
\end{equation}
where  $A_2=\frac{\kappa_2a_1(f_2^{loc})^3}{E_{\max}B}-\frac{b_1}{B}$ and $B_2=\frac{a_1}{E_{\max}B}$. 
\begin{proof}
The derivation is provided in Appendix \ref{OptimalDerivation2}.
\end{proof}
From the above solution, we can directly obtain the optimal solution for two-user case NOMA MEC and the minimum task completion time can be obtained. The complexity is largely reduced compared with the proposed Algorithm \ref{Alg1}. 

\section{Simulation Results}
In this section, the performance of our proposed resource allocation schemes is evaluated by simulation results. In the simulation results, our proposed BSS algorithm and the optimal analysed solution are compared with three benchmark schemes, i.e., orthogonal frequency-division multiple access (OFDMA) based partial scheme, NOMA based full offloading scheme and fully local computing scheme. In the simulation setting up, the users are randomly distributed in a single cell with the radius of 500 m. The channel gain from the $m$-th user to the BS is denoted by $h_{m}=g_{m} \cdot {(1+d_{m}^\alpha)^{-\frac{1}{2}}}$ in which $g_{m}$ is a Rayleigh fading channel coefficient, and $d_m$ is the distance from this user to the BS. The path loss factor $\alpha$ is 3.76. The AWGN power is $\sigma^2=BN_0$ where the AWGN spectral density is $N_0= -174$ dBm/Hz. For the computational resource at each user, we set $C_m=10^3, \forall m$ cycles/bit.
 
\begin{figure}[t]
	\centering
	\includegraphics[width=0.6\linewidth]{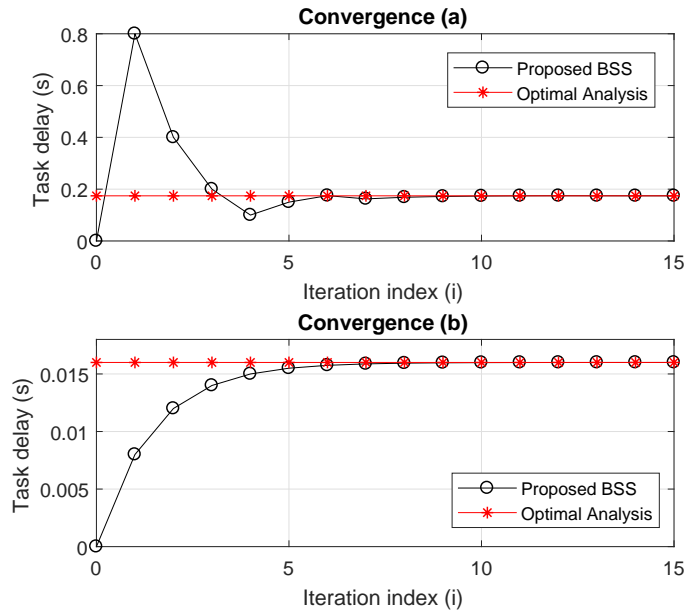}\\
	\caption{The convergence and optimality of Algorithm 1.} \label{Fig1}
\end{figure}

In Fig. \ref{Fig1}, we provide the convergence and optimality of the proposed BSS algorithm. In figure 1 (a), we set the bandwidth to $B=1$ MHz. The length of computation tasks for each user is set to $L_1=L_2=1.6\times10^6$ bits. The maximum energy consumption for each user is $E_{\max}=0.2$ Joule, and the maximum power for each user is  $P_{\max}=0.01$ W, and the effective capacitance coefficient for each CUP cycle of the local users are $\kappa_m=10^{-28}\times[10,1]$. In Fig. 1 (b), we set the bandwidth to $B=1$ Hz, and the length of computation tasks for each user is set to $L_1=L_2=1.6\times10^4$ bits. Since the offloading data rate is small due to small bandwidth, most the tasks are computed by local users. Therefore, the minimum  completion time in BSS algorithm keeps increasing by each iteration until its optimal. From this figure, we can see that our proposed BSS algorithm converges within 10 iterations. The convergence point is perfectly matched with the optimal analytical solution, which is derived by Lagrangian approach.

%\begin{figure}[t]
%	\centering
%	\includegraphics[width=0.6\linewidth]{Fig2_RateConvergence}\\
%	\caption{The offloading data rate convergence with different maximum powers.} \label{Fig2}
%\end{figure}
%
%Figure \ref{Fig2} presents the offloading data rate convergence of the proposed scheme with different maximum powers, where the task length for each user is $L_1=L_2=1.6\times10^6$ bits and $B=1$ MHz. From this figure, we can see that data rate of each user fluctuates at first and then converges within 10 iterations. It also shows that the scheme with larger $P_{\max}$ provides higher data rate than the scheme with lower $P_{\max}$. $U_2$ can always achieve higher data rate than $U_1$. This is because that the BS decodes signals of $U_1$ first according to the SIC decoding order. When the BS decodes the signals from $U_2$, the interference from $U_1$ can be removed.

\begin{figure}[th]
	\centering
	\includegraphics[width=0.6\linewidth]{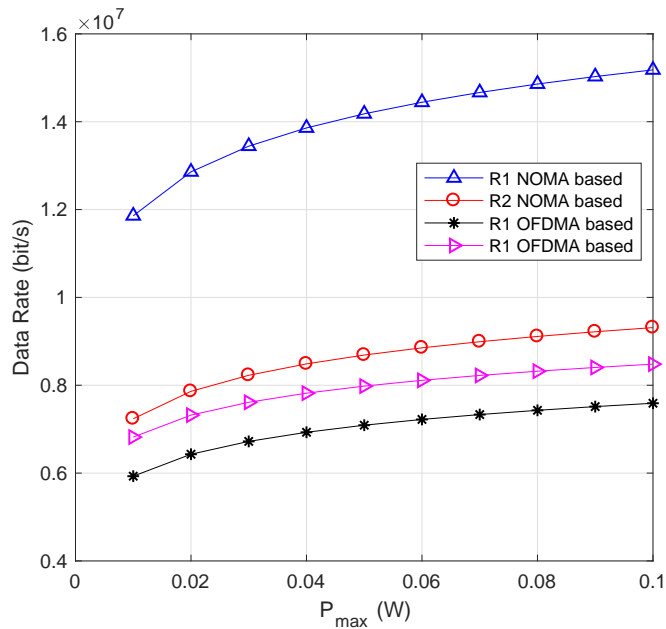}\\
	\caption{The offloading data rate comparison of NOMA based and OFDMA based schemes.} \label{Fig3}
\end{figure}

Figure \ref{Fig3} illustrates the data rate comparison of NOMA offloading and OFDMA offloading schemes versus the maximum power. In this figure, we set $B=1$ MHz and $P_{\max}=0.01$ W. The tasks lengths of these two users are $L_1=1.6\times10^6$ and $L_2=1\times10^6$, respectively. The maximum energy consumption is set to $E_{\max}=0.02$ Joule. In NOMA system, two users offload their tasks on the same frequency band. In OFDMA system, we assume that each user occupies one resource block (frequency band) to offload its task. We can see that offloading data rate increases when the maximum power increases. From this figure, we can see that the scheme with NOMA system can achieve higher data rate than the scheme with the OFDMA system, which demonstrates that NOMA can achieve higher offloading data rate than the OFDMA system.

\begin{figure}[th]
	\centering
	\includegraphics[width=0.6\linewidth]{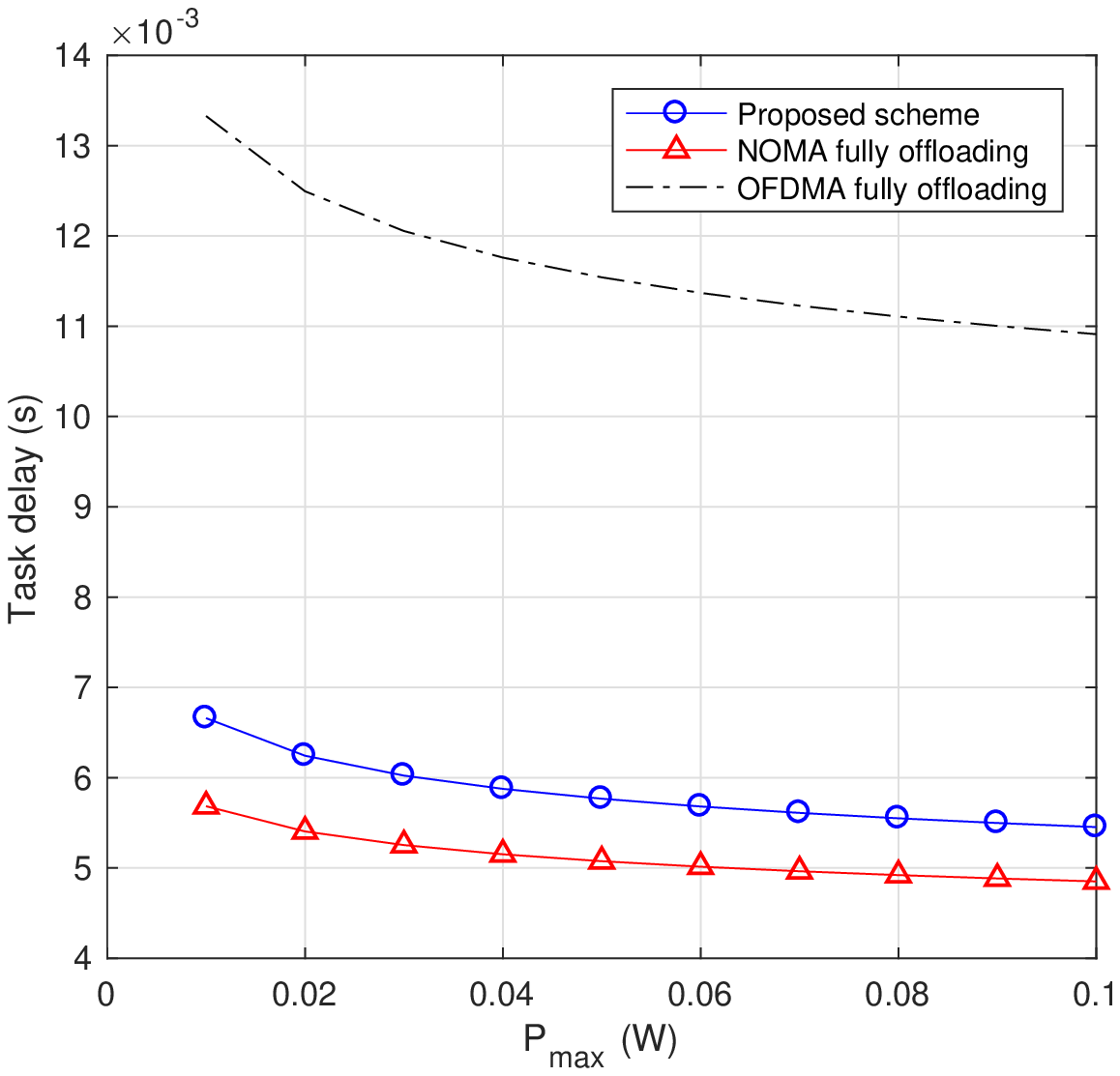}\\
	\caption{The task completion time comparison of different schemes.} \label{Fig4}
\end{figure}

Figure \ref{Fig4} demonstrates the task  completion time comparison between our proposed scheme and the benchmark schemes, OFDMA based offloading scheme, NOMA based fully offloading scheme. Assuming the same task length $L_m=1.6\times10^6$ for both schemes. Our proposed scheme can achieve lower task  completion time than the OFDMA based scheme. This is because NOMA based offloading scheme can provide higher offloading data rate than OFDMA scheme, shown in Fig. \ref{Fig3}. The task completion time of the three schemes decreases when the maximum power increases. This is because larger maximum power range can help provide higher offloading rate, thus less offloading time will be consumed.

\begin{figure}[th]
	\centering
	\includegraphics[width=0.6\linewidth]{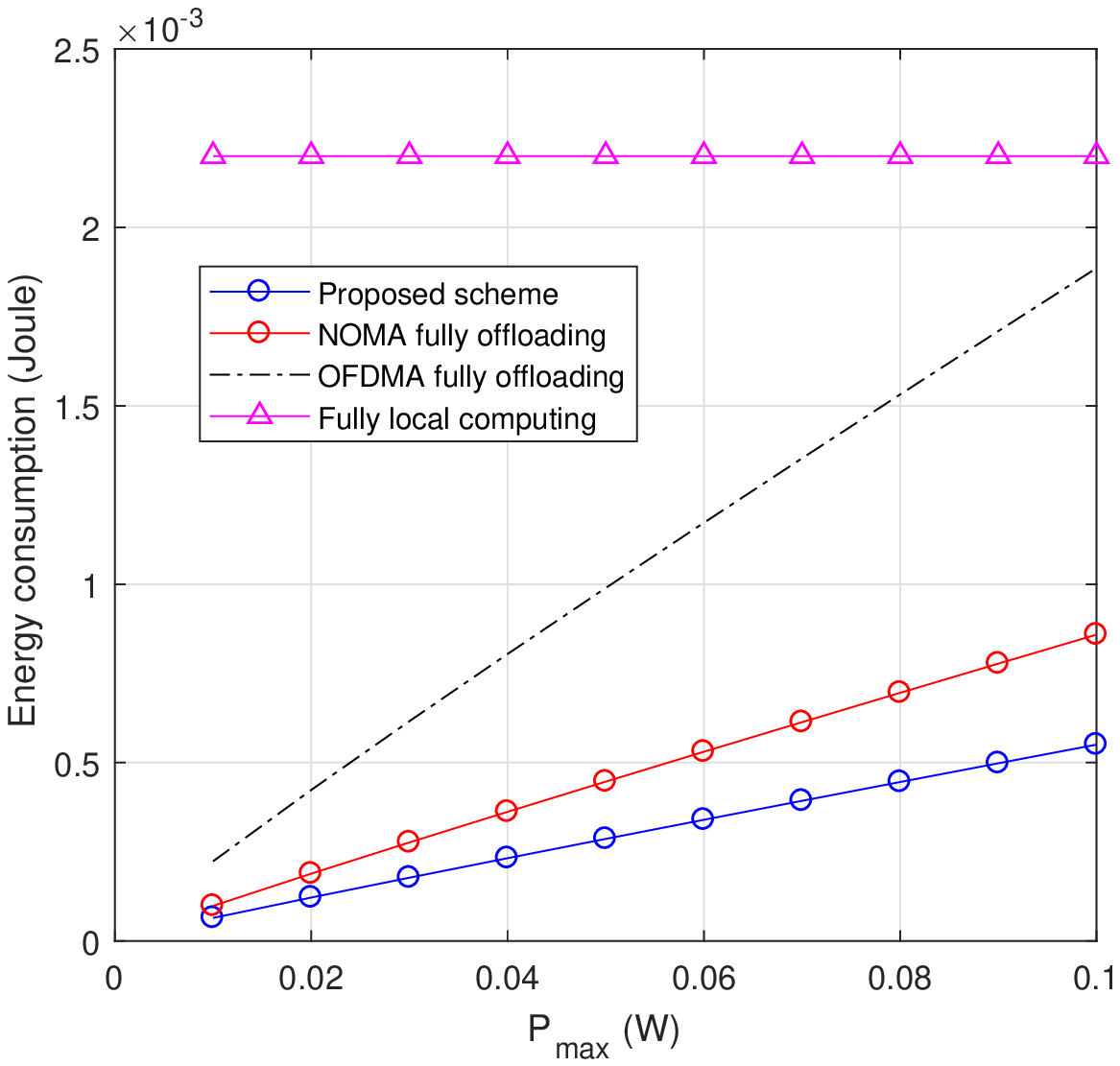}\\
	\caption{The task completion time comparison of different schemes.} \label{Fig5}
\end{figure}

Figure \ref{Fig5} shows the energy consumption comparison of different schemes, including NOMA based fully offloading scheme, NOMA based fully offloading scheme and fully local computing scheme. In this figure, we set $B=1$ MHz and $\kappa_m=2\times10^{-28}$. The tasks lengths of these two users are $L_1=1.6\times10^6$ and $L_2=1\times10^5$, respectively. The CPU frequency is $0.1\times10^{9}$. It can be observed that our proposed NOMA based partial offloading scheme can consume the lowest energy compared to the fully NOMA offloading and OFDMA offloading and fully local computing schemes. The energy consumption increases when the maximum power increases. The energy consumption of local computing scheme stays the same with the increasing maximum power. This is because all that the energy consumption of locally computing scheme is only affected by task length and computing capability of user devices.

\begin{figure}[th]
	\centering
	\includegraphics[width=0.6\linewidth]{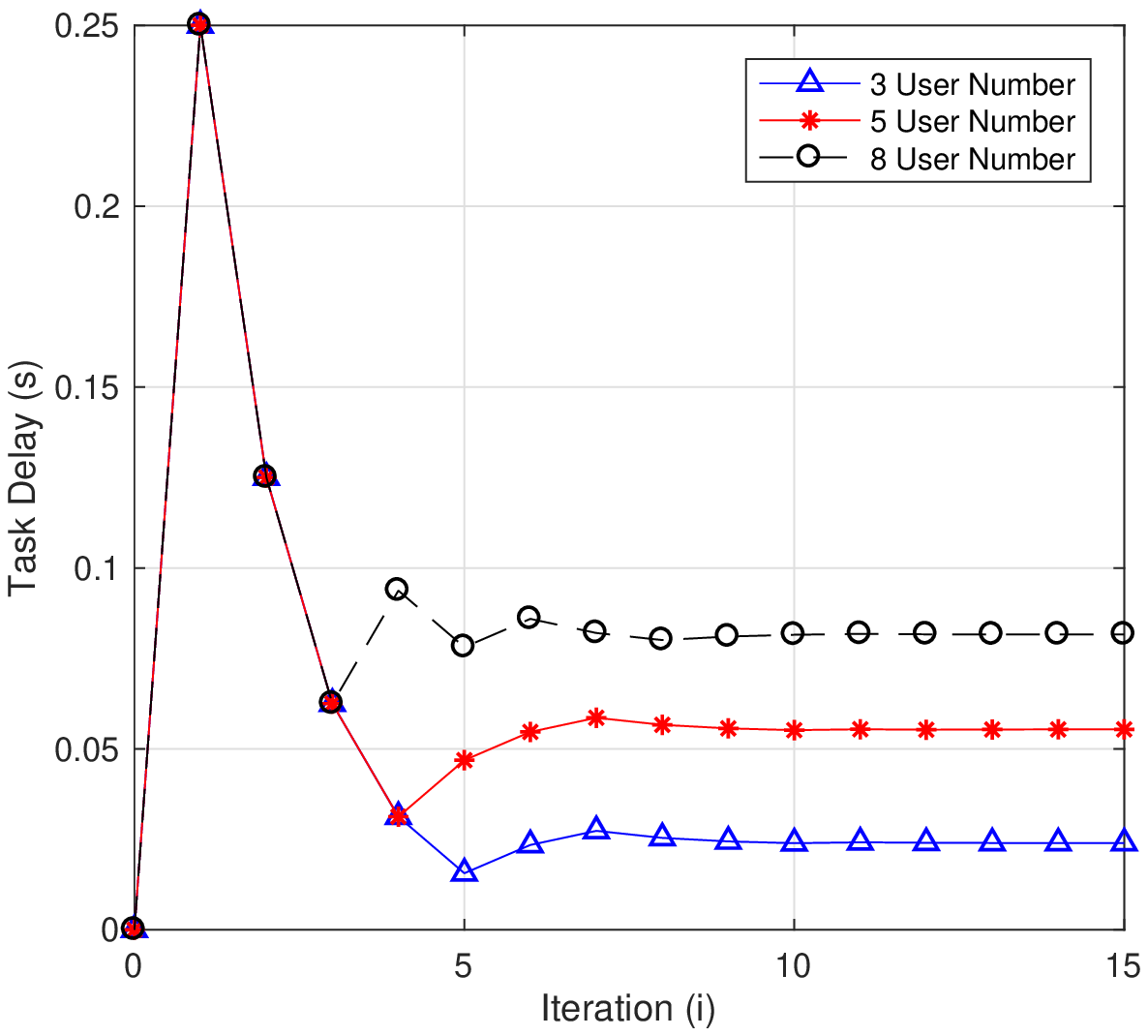}\\
	\caption{The convergence and optimality of Algorithm 1 with multiple users.} \label{Fig6}
\end{figure}

Figure \ref{Fig6} presents the convergence BSS algorithm with multiuser case $M>2$. In this figure, we set $B=1$ MHz and $L_m=1.6\times10^6$ bits, $E_{\max}=0.2$ Joule,  $P_{\max}=0.01$ W, and $\kappa_m=10^{-28}[10,1]$. From this figure, we can see that our proposed BSS algorithm converges within 10 iterations with different user numbers. The scheme with 8 users has the highest completion time than the schemes with 3 and 8 users.

\begin{figure}[th]
	\centering
	\includegraphics[width=0.6\linewidth]{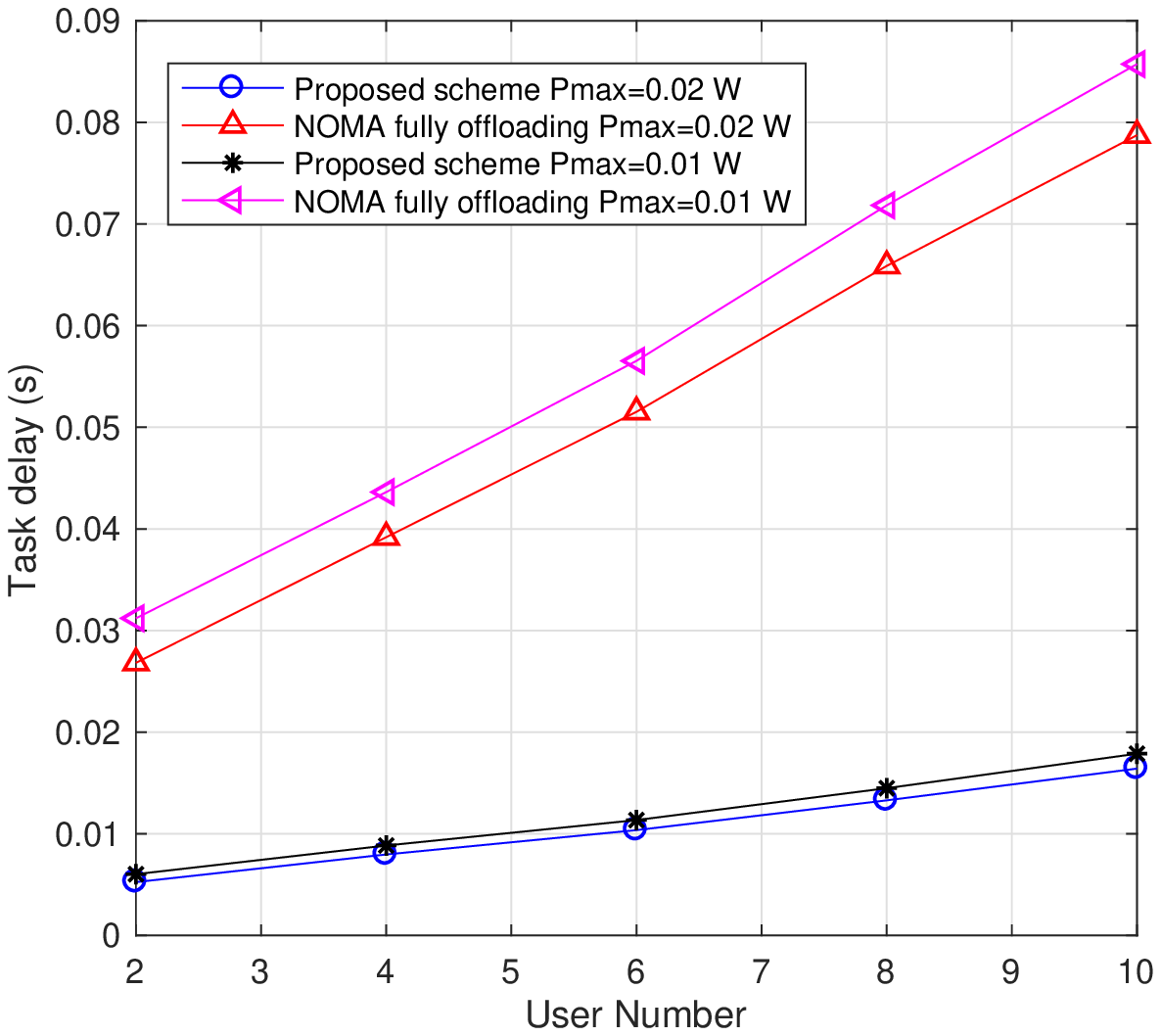}\\
	\caption{The task completion time versus user number with different powers.} \label{Fig7}
\end{figure}

Figure \ref{Fig7} demonstrates the task completion time comparison between NOMA based partial offloading scheme and NOMA fully offloading scheme by considering different offloading power limits. The parameters are set as the same as Fig \ref{Fig6}. In this figure, the task completion time increases when the user number increases. The scheme with higher transmit power limit $P_{\max}=0.02$ W will cost less time than the scheme with $P_{\max}=0.01$ W. It also shows that NOMA based partial offloading scheme can provide better performance than the NOMA based fully offloading scheme. 

\section{Conclusions}
In this paper, we have investigated the optimal task partition and power allocation to minimize task completion time for NOMA multiuser MEC networks. The optimization problem has been formulated as quasi-convex problem. Then BSS algorithm has been proposed to achieve the minimum task completion time, which is the global optimum. To further reduce the complexity of the proposed algorithm, the closed-form optimal power allocation and offloading task ratio expressions have been derived for two-user case via Lagrangian approach based on analytic insights obtained from the analysis. Simulation results have demonstrated the convergence and optimality the proposed schemes, which can provide an effective solution to minimize task completion time for NOMA multiuser MEC networks.
\appendices
\section{Proof of Lemma \ref{Lemma1}} \label{Lemma1Proof}
%In this system, $M$ users offload their tasks to the MEC server  simultaneously through one subchannel. By applying the SIC technique at the receiver, the BS equipped with the MEC server decodes the signals in decreasing order of their channel gains,  
%    $|h_{1}|^2\leq |h_{2}|^2 \leq \cdots \leq |h_{M}|^2$. 
We assume that $M$ users transmit their tasks in the same period $T$, which indicates:
%    \begin{equation}
%       T=T_{1}^{off}=T_{2}^{off}=\cdots=T_{M}^{off} 
%    \end{equation}
%   and
   \begin{equation}\label{T1=T2=T}
       T=\frac{\beta_{1}L_1}{R_{1}} =\frac{\beta_{2}L_2}{R_{2}}=\cdots=\frac{\beta_{M}L_M}{R_{M}}.
    \end{equation}
    Since $\frac{a_1}{b_1}=\frac{a_2}{b_2}=\cdots=\frac{a_M}{b_M}$ can be written by $\frac{a_1}{b_1}=\frac{e_2a_1}{e_2b_1}=\cdots=\frac{e_Ma_1}{e_Mb_1}$ where $e_m=\frac{a_m}{a_1}=\frac{b_m}{b_1}$. Thus $\frac{a_1}{b_1}=\frac{a_1(1+e_2+e_3+\cdots+e_M)}{b_1(1+e_2+e_3+\cdots+e_M)}=\frac{a_1+a_2+\cdots+a_M}{b_1+b_2+\cdots+b_M}$. Therefore,
     \begin{equation}
       T=\frac{\beta_{1}L_1+\beta_{2}L_2+\cdots+\beta_{m}L_m}{R_{1}+R_{2,n}+\cdots+R_{m}} =\frac{\sum\limits_{i=1}^m\beta_{i}L_i}{\sum\limits_{i=1}^mR_{m}}, \forall m.
    \end{equation} 
    The offloading sum rate can be derived as 
    \begin{equation}
    \begin{aligned}
    	 \sum\limits_{i=1}^mR_{i}=&B\log_2(\sigma^2+|h_{1}|^2p_{1})+B\log_2\left(\frac{\sigma^2+|h_{1}|^2p_{1}+|h_{2}|^2p_{2}}{\sigma^2+|h_{1}|^2p_{1}}\right)+\cdots \\
    	 &+B\log_2\left(\frac{\sigma^2+\sum\limits_{i=1}^m|h_{i}|^2p_{i}}{\sigma^2+\sum\limits_{i=1}^{m-1}|h_{i}|^2p_{i}}\right)
    	 =B\log_2\left(\sigma^2+\sum\limits_{i=1}^m|h_{i}|^2p_{i}\right).\\
    \end{aligned}
    \end{equation}
    Define the number of transmitted bits as $L^{off}=\sum\limits_{m=1}^M \beta_{m}L_m$, and $R^{off}=B\log_2\left(\sigma^2+\sum\limits_{m=1}^M|h_{m}|^2p_{m}\right)$. Then the offloading time can be written by
    \begin{equation}\label{Tm_off}
       T=\frac{\sum\limits_{i=1}^m\beta_{i}L_i}{B\log_2\left(\sigma^2+\sum\limits_{i=1}^m|h_{i}|^2p_{i}\right)}, \forall m.
    \end{equation} 
    Now let us prove from \eqref{Tm_off} to \eqref{T1=T2=T}.
    When $m=1$, we have $T=T_1^{off}$.
    When $m=2$, we have $T=\frac{\sum\limits_{i=1}^2\beta_{i}L_i}{\sum\limits_{i=1}^2R_{2}}$.
    Since $\frac{\beta_1L_1}{R_1}=\frac{\beta_1L_1+\beta_2L_2}{R_1+R_2}$, we can have  $\frac{\beta_1L_1}{R_1}=\frac{\beta_2L_2}{R_2}$.
    By deduction, we can have \eqref{T1=T2=T}.
    We finish the proof of Lamma \ref{Lemma1}.
    
    \section{Proof of Proposition \ref{T1=T2}}\label{T_1_off=T_m_off}
    According to Lemma \ref{Lemma1}, the offloading time minimization problem of different users can be represented by 
    \begin{equation}
    \begin{aligned} \label{Problem_proof2}
    %&\mathop {\min \quad \max }\limits_{\{\beta_1,p_1\}}\left\{T_1^{off},T_1^{loc}\right\}\\
    &\mathop {\min \quad \max }\limits_{\{\bm{\beta},\bm{p}\}}\left\{T_1^{off},T_2^{off}, \cdots,T_M^{off}\right\}.
    \end{aligned}
    \end{equation}
    Note that the minimum latency is  $T^*$ with the optimal solution $\{\bm{\beta}^*,\bm{p}^*\}$. This optimal solution is only obtained when $T_m^{off}=T_{m'}^{off},\forall m\neq m'$. Proof by contradiction can be exploited to prove this proposition.

    Assume that $U_{m'}$ decodes its signal firstly, and the optimal solution is obtained when $T_m^{off}>T_{m'}^{off}, m<m'$. Thus the minimum latency is $T^*$ with the optimal solution $\{\beta_m^*,\beta_{m'}^*,p_m^*,p_{m'}^*\}$. In this case, if we increase $p_m^*$ to $\hat{p}_m$, then $T_m^{off}$ will be decreased. Since $p_{m'}^*$ is fixed, and $T_{m'}^{off,*}$ will be increased. Therefore, there must exist $\hat{p}_m$ satisfying $T_m^{off}(p_m^*,p_{m'}^*)>\hat{T}_m^{off}(\hat{p}_m,p_{m'}^*)=\hat{T}_{m'}^{off}(\hat{p}_m,p_{m'}^*)>T_{m'}^{off}(p_m^*,p_{m'}^*)$. Therefore,  $\hat{T}_m^{off}(\hat{p}_m,p_{m'}^*)=\hat{T}_{m'}^{off}(\hat{p}_m,p_{m'}^*)$ should be the optimal time consumption since it has lower value than $T_m^{off}(\beta_m^*,p_m^*)$. This contradicts the assumption that  $\{\beta_m^*,p_m^*\}$ is the optimal solution to problem \eqref{Problem_proof2}.
    
    Assume that the optimal solution is obtained when $T_m^{off}<T_{m'}^{off}, m<m'$. Thus the minimum latency is $T^*$ with the optimal solution $\{\beta_m^*,\beta_{m'}^*,p_m^*,p_{m'}^*\}$. In this case, if we increase $p_m^*$ to $\hat{p}_m$, then $T_m^{off}$ will be decreased. Since $p_{m'}^*$ is fixed, and $T_{m'}^{off,*}$ will be increased. Therefore, there must exist $\hat{p}_m$ satisfying $T_m^{off}(p_m^*,p_{m'}^*)<\hat{T}_m^{off}(\hat{p}_m,p_{m'}^*)=\hat{T}_{m'}^{off}(\hat{p}_m,p_{m'}^*)<T_{m'}^{off}(p_m^*,p_{m'}^*)$. Therefore,  $\hat{T}_m^{off}(\hat{p}_m,p_{m'}^*)=\hat{T}_{m'}^{off}(\hat{p}_m,p_{m'}^*)$ should be the optimal time consumption since it has lower value than $T_{m'}^{off}(\beta_m^*,p_m^*)$. This contradicts the assumption that  $\{\beta_m^*,p_m^*\}$ is the optimal solution to problem \eqref{Problem_proof2}.
    
    Therefore, it can be concluded that the optimal solution to problem \eqref{Problem_proof2} can only be obtained when $T_m^{off}=T_{m'}^{loc}, \forall m\neq m'$.
    
\section{The Proof of Proposition \ref{quasi}} \label{QuasiProof}
The $\alpha$ sublevel sets of $\max \limits_{\{\beta_{m},p_{m}\}}  \left\{\tilde{T}_m^{off},T_m^{loc}, \forall m \right\}$ can be defined as 
\begin{equation}
S_{\alpha_T}=\{\{\beta_{m},p_{m}\}|\max{\left\{\tilde{T}_m^{off},T_m^{loc},\forall m \right\}}\leq \alpha_T\}.
\end{equation}
According to the definition of quasi-convex function \cite{BoydConv2004}, $\max{\left\{\tilde{T}_m^{off},T_m^{loc},\forall m\right\}}$ is quasi-convex if and only if its sublevel sets $S_{\alpha_T}$ is convex for any $\alpha_T$. Let us prove the $S_{\alpha_T}$ is convex set first. In our system, when $\alpha_T\leq0$, there are no solutions satisfying $\max{\left\{\tilde{T}_m^{off},T_m^{loc},\forall m \right\}}\leq\alpha_T$. When $\alpha_T> 0$, $S_{\alpha_T}$ can be rewritten as the following inequality:
\begin{equation}
	\frac{\sum\limits_{i=1}^m\beta_iL_i}{\log_2(1+\sum\limits_{i=1}^m|h_i|^2p_i)}\leq \alpha_T, \forall m,\quad
	\frac{(1-\beta_m)L_mC_m}{f_m^c}\leq\alpha_T, \forall m \label{T_1_c}.
\end{equation}
It can be observed that \eqref{T_1_c} are linear equalities respect to $\beta_1$ and $\beta_2$, respectively. Therefore, \eqref{T_1_c} is convex sets for any $m$. To prove $S_{\alpha_T}$ is strictly quasi-convex, we need to prove that $\sum\limits_{i=1}^m\beta_{i}L_{i}-\log_2(\sigma^2+\sum\limits_{i=1}^m)$ is a strictly convex function. Since the first term $\sum\limits_{i=1}^m\beta_{i}L_{i}$ is liner function and the second term $-\log_2(\sigma^2+\sum\limits_{i=1}^m)$ is convex function because $-\log(\cdot)$ is convex function and the term inside of $-\log$ is a liner function. Let us take two user case for example. The function in \eqref{T_1_c}, $(\beta_1L_1+\beta_2L_2)/\alpha_T-\log_2(1+|h_1|^2p_1^{off}+|h_2|^2p_2^{off})$, is a strictly convex function. To simplify the analysis, we use $a_1$, $a_2$, $b_1$ and $b_2$ to represent $L_1\alpha_T$, $L_2\alpha_T$, $|h_1|^2$ and $|h_2|^2$, respectively. We also use the variables $x_m$ and $y_m$ to respectively present $\beta_m$ and $p_m$ where $m=1, 2$. Thus the function in \eqref{T_1_c} can be rewritten as $f(\bm{x},\bm{y})=(a_1x_1+a_2x_2)-\log_2(1+b_1y_1+b_2y_2)$. To prove its convexity, we need to prove that its Hessian matrix is positive semi-definite matrix. The Hessian matrix of $f(\bm{x},\bm{y})$ is:
\begin{equation} 
\begin{aligned}
\mathbf{H}(f)=
\frac{1}{(1+b_1y_1+b_2y_2)^2\ln2}\left[
\begin{matrix}
0&0&0&0\\
0&0&0&0\\
0 &0&b_1^2 & b_1b_2 \\
0&0&b_1b_2 & b_2^2.
\end{matrix}
\right]
\end{aligned}
\end{equation}
According to the definition of semi-definite, we first define a non-zero column vector $\bm{v}=[v_1,v_2,v_3,v_4]$. To prove $\mathbf{H}$ is positive semi-definite matrix, we need to prove $\boldsymbol{v}^T\mathbf{H}\boldsymbol{v}\geq 0$. We have

\begin{equation} 
\begin{aligned}
\boldsymbol{v}^T\mathbf{H}(f)\boldsymbol{v}=
\frac{(b_1v_3+b_2v_4)^2}{(1+b_1y_1+b_2y_2)^2\ln2}.
\end{aligned}
\end{equation}
Since $v_3$, $v_4$, $b_1>0$ and $b_2>0$ have non-zero values, $\boldsymbol{v}^T\mathbf{H}(f)\boldsymbol{v}>0$. Therefore, Hessian matrix of $f$ is positive definite matrix. Hence \eqref{T_1_c} is a convex set. Since all the sublevel sets are convex, the intersection of these convex sets is convex. Similarly, we could prove the constraint \eqref{eq:E_1_con} is convex set. Therefore, it can be concluded that problem \eqref{Prob:T_min_2} is a quasiconvex problem.

 \section{Proof of Proposition \ref{T_loc=T_off}}\label{T_m_loc=T_m_off}
 The task completion time minimization problem for each user, i.e., $U_m$, can be represented as
 \begin{equation}
 \begin{aligned} \label{Problem_proof}
 \mathop {\min \quad \max }\limits_{\{\beta_m,p_m\}}\left\{T_m^{off}, T_m^{loc}\right\}.
 \end{aligned}
 \end{equation}
 Note that the optimal solution to the above problem is $\{\beta_m^*,p_m^*\}$. This optimal solution is only obtained when $T_m^{off}=T_m^{loc}$. Proof by contradiction can be used to prove this proposition.
 
 Assume that the optimal solution is obtained when $T_m^{off}(\beta_m^*,p_m^*)> T_m^{loc}(\beta_m^*)$. Note that $E_m^{off}$ will increase and $E_m^{loc}$ will decrease when $\beta_m$ increases. Thus the optimal solution can be written as $T_m^{off}(\beta_m^*,p_m^*)$ satisfying the energy constraint and power constraint. In this case, there must exist $\hat{\beta}<\beta_m$ satisfying $T_m^{off}(\beta_m^*,p_m^*)>\hat{T}_m^{off}(\hat{\beta}_m,\hat{p}_m)=\hat{T}_m^{loc}(\hat{\beta}_m,\hat{p}_m)> T_m^{loc}(\beta_m^*)$. Therefore,  $\hat{T}_m^{off}(\hat{\beta}_m,\hat{p}_m)=\hat{T}_m^{loc}(\hat{\beta}_m,\hat{p}_m)$ should be the optimal time consumption since it has lower value than $T_m^{off}(\beta_m^*,p_m^*)$. This contradicts the assumption that $T^*$ is the minimum latency and $\{\beta_m^*,p_m^*\}$ is the optimal solution to problem \eqref{Problem_proof}.
 
 Similarly, assume that the optimal solution is obtained when $T_m^{off}(\beta_m^*,p_m^*)< T_m^{loc}(\beta_m^*)$. Note that $E_m^{off}$ will increase and $E_m^{loc}$ will decrease when $\beta_m$ increases. Thus the optimal solution can be written as $T_m^{off}(\beta_m^*,p_m^*)$ satisfying the energy constraint and power constraint. In this case, there must exist $\hat{\beta}>\beta_m$ satisfying $T_m^{off}(\beta_m^*,p_m^*)<\hat{T}_m^{off}(\hat{\beta}_m,\hat{p}_m)=\hat{T}_m^{loc}(\hat{\beta}_m,\hat{p}_m)< T_m^{loc}(\beta_m^*)$. Therefore,  $\hat{T}_m^{off}(\hat{\beta}_m,\hat{p}_m)=\hat{T}_m^{loc}(\hat{\beta}_m,\hat{p}_m)$ should be the optimal time consumption since it has lower value than $T_m^{off}(\beta_m^*,p_m^*)$. This contradicts the assumption that $\{\beta_m^*,p_m^*\}$ is the optimal solution to problem \eqref{Problem_proof}.
 
 Therefore, it can be concluded that the optimal solution to problem \eqref{Problem_proof} can only be obtained when $T_m^{off}=T_m^{loc}, \forall m$.

\section{Transformation from problem \eqref{Prob:T_min_2_2} to problem \eqref{Prob:p1,p2} }  \label{ObjTransfer}
From \eqref{equalities}, we have
\begin{subequations}
	\begin{align}
		&(1-\beta_1)L_1C_1f_2^{loc}=(1-\beta_2)L_2C_2f_1^{loc}\\
		&\beta_2L_2=L_2-\frac{L_1C_1f_2^{loc}}{C_2f_1^{loc}}+\beta_1\frac{L_1C_1f_2^{loc}}{C_2f_1^{loc}}.
	\end{align}
\end{subequations}
and 
\begin{equation}\label{b1l1}
\beta_1L_1+\beta_2L_2=\beta_1L_1A_1+B_1
\end{equation}  
where $A_1=1+\frac{C_1f_2^{loc}}{C_2f_1^{loc}}$ and $L_2-L_1\frac{C_1f_2^{loc}}{C_2f_1^{loc}}$.
Since 
\begin{equation}
\frac{\beta_1L_1+\beta_2L_2}{R}=\frac{(1-\beta_1)L_1C_1}{f_1^{loc}},
\end{equation}   
we have 
\begin{equation}
\beta_1L_1=\frac{L_1C_1R-B_1f_1^{loc}}{A_1f_1^{loc}+C_1R}.
\end{equation}  
According to \eqref{b1l1}, we have 
\begin{equation}
\beta_1L_1+\beta_2L_2=\frac{L_1C_1A_1R+B_1C_1R}{A_1f_1^{loc}+C_1R}.
\end{equation} 
and 
 \begin{equation}
 	\frac{\beta_1L_1+\beta_2L_2}{R}=\frac{L_1C_1A_1+B_1C_1}{A_1f_1^{loc}+C_1R}=\frac{L_1+L_2}{f_1^{loc}/C_1+f_2^{loc}/C_2+R}.
 \end{equation}
Therefore, problem \eqref{Prob:T_min_2_2} is transformed to problem \eqref{Prob:p1,p2}.

\section{Convexity proof of Problem \eqref{Prob:p1,p2} } \label{ConvexityProof}
In problem \eqref{Prob:p1,p2}, the objective function can be rewritten by  $f(p_1,p_2)=\frac{a_1}{b_1+B\log_2(1+|h_1|^2p_1+|h_2|^2p_2)}$ where $a_1=L_1+L_2$ and $b_1=\frac{f_1^{loc}}{C_1}+\frac{f_2^{loc}}{C_2}$. To simplify the analysis, we rewrite the objective function as $f(p_1,p_2)=\frac{a_1}{b_1+R}$ where $R=B\log_2(1+|h_1|^2p_1+|h_2|^2p_2)$. Note that $f(p_1,p_2)$ is convex function with respect to $p_1$ and $p_2$ according to the convexity properties. 

Now let us prove the convexity of the constraint \eqref{eq:E1} and \eqref{eq:E2}. Since $b_1+R>0$, the constraint \eqref{eq:E1} can be rewritten by 
\begin{subequations}
	\begin{align}
		\kappa_1a_1(f_1^{loc})^3+a_1p_1- E_{\max}(b+R)\leq 0
	\end{align}
\end{subequations}
Let $F(p_1,p_2)=\kappa_1a_1(f_1^{loc})^3+a_1p_1- E_{\max}(b+R)$. To prove its convexity, we need to prove that its Hessian matrix is positive definite matrix. The Hessian matrix of $F(p_1,p_2)$ is:
\begin{equation} 
\begin{aligned}
\mathbf{H}(F)=
\Phi \left[
\begin{matrix}
|h_1|^4 &|h_1|^2|h_2|^2 \\
|h_1|^2|h_2|^2&|h_2|^4
\end{matrix}
\right]
\end{aligned}
\end{equation}
where $\Phi=\frac{E_{\max}}{(1+|h_1|^2p_1+|h_2|^2p_2)^2\ln2}$.
Hessian matrix of a convex function is positive semi-definite. According to the definition of semi-definite, we first define a non-zero column vector $\bm{v}=[v_1,v_2]$. To prove $\mathbf{H}$ is positive definite matrix, we need to prove $\boldsymbol{v}^T\mathbf{H}\boldsymbol{v}>0$. We have

\begin{equation} 
\begin{aligned}
\boldsymbol{v}^T\mathbf{H}(F)\boldsymbol{v}=\Phi 
(v_1|h_1|^2+v_2|h_2|^2)^2.
\end{aligned}
\end{equation}
Since $v_1$, $v_2$ have non-zero values and $\Phi >0$, $\boldsymbol{v}^T\mathbf{H}(F)\boldsymbol{v}>0$. Hence, \eqref{eq:E1} is a convex set. Similarly, we could prove the constraint \eqref{eq:E2} is convex set.
Therefore, problem \eqref{Prob:p1,p2} is convex.

\section{Derivation of the Optimal Solution to Problem \eqref{Prob:p1,p2} }\label{OptimalDerivation2}
%We rewrite problem \eqref{Prob:p1,p2} as follows
%\begin{subequations}\label{Prob:p1}
%	\begin{align}
%	\mathop {\min }\limits_{ \{p_{1},p_{2}\}} \quad &\frac{a_1}{b_1+B\log_2(1+|h_1|^2p_1+|h_2|^2p_2)}\\
%	\text{s.t.}  \quad	&0\leq p_{1}\leq P_{\max},0\leq p_{2}\leq P_{\max},  \label{p_range2}\\\nonumber
%		&\kappa_1(1- \beta_1) L_1C_1 (f_1^{loc})^2\\
%	&+\frac{\beta_1L_1+\beta_2L_2}{\log_2(1+|h_1|^2p_1+|h_2|^2p_2)}p_1 \leq E_{\max}, \label{eq:E12}\\\nonumber
%	&\kappa_2(1- \beta_2) L_2C_2 (f_2^{loc})^2 \\
%	&+\frac{\beta_1L_1+\beta_2L_2}{\log_2(1+|h_1|^2p_1+|h_2|^2p_2)}p_2 \leq E_{\max}\label{eq:E22}\\  \nonumber 
%	\end{align}
%\end{subequations}
In problem \eqref{Prob:p1,p2}, let $a_1=L_1+L_2$ and $b_1=\frac{f_1^{loc}}{C_1}+\frac{f_2^{loc}}{C_2}$.
The Lagrangian function of problem \eqref{Prob:p1,p2} can be written as:
\begin{equation}
\begin{aligned}
\mathcal{L}=&\frac{a_1}{b_1+B\log_2(1+|h_1|^2p_1+|h_2|^2p_2)}+\lambda_1(-p_1)+\lambda_2(p_1-P_{\max})+\lambda_3(-p_2)+\lambda_4(p_2-P_{\max})\\
&+\lambda_{5}\left(\kappa_1a_1(f_1^c)^3+a_1p_1-E_{\max}(b_1+R(p_1,p_2))\right)\\
&+\lambda_{6}\left(\kappa_1a_1(f_1^c)^3+a_1p_2-E_{\max}(b_1+R(p_1,p_2))\right)
\end{aligned}
\end{equation}
where $\lambda_i$ are Lagrangian multipliers corresponding to constraints in problem \eqref{Prob:p1,p2}. Since it is convex problem and satisfies Slater's condition, thus KKT conditions are necessary and sufficient to the optimal solution \cite{BoydConv2004}. To obtain the optimal solution,  KKT conditions (i.e., stationary condition, primal feasible condition, dual feasibility condition and Complementary slackness condition)  of problem \eqref{Prob:p1,p2} can be written as follows:
Stationary condition:
\begin{equation}
	\begin{aligned}
	\frac{\partial\mathcal{L} }{\partial p_1}=&\frac{-a_1\frac{B|h_1|^2}{(1+|h_1|^2p_1^*+|h_2|^2p_2^*)\ln(2)}}{\left(b_1+B\log_2(1+|h_1|^2p_1^*+|h_2|^2p_2^*)\right)^2}-\lambda_1+\lambda_2\label{partial_p1}\\ %\label{partial_p1}\\
	&+\lambda_5\left(a_1-\frac{E_{\max}B|h_1|^2}{(1+|h_1|^2p_1^*+|h_2|^2p_2^*)\ln(2)}\right)=0
	\end{aligned}
\end{equation}
\begin{equation}
\begin{aligned}
\frac{\partial\mathcal{L} }{\partial p_2}=&\frac{-a_1\frac{B|h_2|^2}{(1+|h_1|^2p_1^*+|h_2|^2p_2^*)\ln(2)}}{\left(b_1+B\log_2(1+|h_1|^2p_1^*+|h_2|^2p_2^*)\right)^2}-\lambda_3+\lambda_4\label{partial_p2}\\ %\label{partial_p1}\\
&+\lambda_6\left(a_1-\frac{E_{\max}B|h_2|^2}{(1+|h_1|^2p_1^*+|h_2|^2p_2^*)\ln(2)}\right)=0
\end{aligned}
\end{equation}
Primal feasible condition:
\begin{subequations}
	\begin{align}
	&-p_1^*\leq 0,\ -p_2^*\leq 0\\
	&p_1^*-P_{\max}\leq 0,\ p_2^*-P_{\max}\leq 0\label{lambdaPmax1}\\
	&\kappa_1a_1(f_1^{loc})^3+a_1p_1^* -E_{\max}(b_1+B\log_2(1+|h_1|^2p_1^*+|h_2|^2p_2^*))\leq 0\label{primalE1} \\
    &\kappa_2a_1(f_1^{loc})^3+a_1p_2^*-E_{\max}(b_1+B\log_2(1+|h_1|^2p_1^*+|h_2|^2p_2^*))\leq 0\label{primalE2} \\\nonumber
	\end{align}
\end{subequations}
Dual feasibility condition:
\begin{equation}
\lambda_i \geq 0, i={1, \cdots, 6}.
\end{equation}
Stationary condition:
\begin{subequations}
	\begin{align}
	&\lambda_1p_1^*= 0,\ \lambda_3p_2^*= 0\label{lambdap1}\\
	&\lambda_{2}\left(p_1^*-P_{\max}\right)=0,\ \lambda_{4}\left(p_2^*-P_{\max}\right)=0 \label{lambdaPmax}\\
	&\lambda_{5}\left[\kappa_1a_1(f_1^{loc})^3+a_1p_1^* -E_{\max}\left(b_1+B\log_2\left(1+|h_1|^2p_1^*+|h_2|^2p_2^*\right)\right) \right]= 0 \label{lambdaE1}\\
	&\lambda_{6}\left[\kappa_2a_1(f_1^{loc})^3+a_1p_2^*  -E_{\max}\left(b_1+B\log_2\left(1+|h_1|^2p_1^*+|h_2|^2p_2^*\right)\right) \right]= 0 \label{lambdaE2}\\ \nonumber
	\end{align}
\end{subequations}

In our solution, we have $p_1^*>0$ and $p_2^*>0$. To satisfy \eqref{lambdap1}, we can obtain $\lambda_{1}=0$ and $\lambda_{3}=0$. To satisfy \eqref{partial_p1}, when $a_1\geq \frac{E_{\max}B|h_2|^2}{(1+|h_1|^2p_1^*+|h_2|^2p_2^*)\ln(2)}$, we have at lease one of $\lambda_2$ and $\lambda_5$ is larger than zero. There are three cases to satisfy this condition: 1. $\lambda_2>0, \lambda_5=0$; 2. $\lambda_2=0, \lambda_5>0$; 3. $\lambda_2>0, \lambda_5>0$. when $a_1<\frac{E_{\max}B|h_2|^2}{(1+|h_1|^2p_1^*+|h_2|^2p_2^*)\ln(2)}$Let us first investigate these three cases. 

When $\lambda_2>0,\lambda_5=0$, to satisfy \eqref{lambdaPmax}, we have 
\begin{equation}
p_1^*=P_{\max}. \label{optimal p1=pmax}
\end{equation}
Based on $p_1^*$, to obtain $p_2^*$, we have the following calculation steps.
To satisfy \eqref{partial_p2}, due to the negativity of the first term, we have at least one of $\lambda_4$ and $\lambda_{6}$ is larger than zero based on $a_1>\frac{E_{\max}B|h_2|^2}{(1+|h_1|^2p_1^*+|h_2|^2p_2^*)\ln(2)}$. When $a_1\leq \frac{E_{\max}B|h_2|^2}{(1+|h_1|^2p_1^*+|h_2|^2p_2^*)\ln(2)}$, we must have $\lambda_{4}>0$. To be concluded, we have three	 cases to obtain $p_2^*$: 1. $\lambda_4>0, \lambda_6=0$; 2. $\lambda_4=0, \lambda_6>0$ based on $a_1\geq\frac{E_{\max}B|h_2|^2}{(1+|h_1|^2p_1^*+|h_2|^2p_2^*)\ln(2)}$; 3. $\lambda_4>0, \lambda_6>0$.

\emph{\underline{Case 1}}: $\lambda_2>0,\lambda_5=0$, $\lambda_4>0, \lambda_6=0$, to satisfy \eqref{lambdaPmax}, we have 
\begin{equation}
p_2^*=P_{\max}. \label{optimal p2=pmax}
\end{equation}
In this case, to satisfy the feasible conditions \eqref{primalE1} and \eqref{primalE2}, we must have
\begin{equation}
p_1^*\leq P_{1,w}(p_2^*=P_{\max})
\end{equation} 
where $P_{1,w}(p_2^*)$ is defined in \eqref{p1range_p1}.
\begin{equation}
p_2^*\leq P_{2,w}(p_1^*=P_{\max})
\end{equation}
where $P_{1,w}(p_2^*)$ is defined in \eqref{p2range_p1}.
Therefore, we have the optimal solution \eqref{Case1}.
%\begin{equation}
%\left\{
%\begin{aligned}
%&p_1^*= P_{\max} \quad  &\\
%&p_2^*=P_{\max} \quad   &
%\end{aligned}
%\right.
%\end{equation}
%based on the condition:
%\begin{equation}
%\left\{
%\begin{aligned}
%& P_{\max}\leq P_{1,w}(p_2^*=P_{\max}) \quad  & \\
%& P_{\max} \leq P_{2,w}(p_1^*=P_{\max}). &
%\end{aligned}
%\right.
%\end{equation}
%We need to check if the initial condition $\alpha_T^*<\min\{\frac{L_1C_1}{f_1^{loc}},\frac{L_2C_2}{f_2^{loc}}\}$ can be satisfied.

\emph{\underline{Case 2}}: $\lambda_2>0,\lambda_5=0$, $\lambda_4=0$ and $\lambda_6>0$, to satisfy \eqref{lambdaE2}, we have 
\begin{equation}
\begin{aligned}
\kappa_2a_1(f_2^{loc})^3+a_1p_2^*  - E_{\max}\left(b_1+B\log_2\left(1+|h_1|^2p_1^*+|h_2|^2p_2^*\right)\right) = 0
\end{aligned}
\end{equation}
Thus we have
\begin{equation}
p_2^*=P_{2,w}(p_1^*=P_{\max}). \label{optimal p2=W}
\end{equation}
In this case, we need to satisfy primal feasible condition as $ P_{2,w}(p_1^*=P_{\max})\leq P_{\max}$
Therefore, we can have the optimal solution \eqref{Case2}.
%\begin{equation}
%\left\{
%\begin{aligned}
%&p_1^*= P_{\max} \quad  &\\
%&p_2^*=P_{2,w}(p_1^*=P_{\max}) \quad   &
%\end{aligned}
%\right.
%\end{equation}
%based on the condition:
%\begin{equation}
%\left\{
%\begin{aligned}
%& P_{\max}\leq P_{1,w}(p_2^*=P_{\max}) \quad  & \\
%& P_{\max} \geq P_{2,w}(p_1^*=P_{\max}). &
%\end{aligned}
%\right.
%\end{equation}
 
When $\lambda_2>0,\lambda_5=0$, $\lambda_4>0$ and $\lambda_6>0$, to satisfy \eqref{lambdaPmax} and \eqref{lambdaE2}, we have \eqref{optimal p2=pmax}=\eqref{optimal p2=W}, which leads to
\begin{equation}
P_{\max}=P_{2,w}(p_1^*=P_{\max}). 
\end{equation}	
This case can be included into the above cases, thus this case can be ignored.
	
When $\lambda_2=0, \lambda_5>0$, to satisfy \eqref{lambdaE1}, we have 
\begin{equation}
\begin{aligned}
      \kappa_1a_1(f_1^{loc})^3+a_1p_1^* - E_{\max}\left(b_1+B\log_2\left(1+|h_1|^2p_1^*+|h_2|^2p_2^*\right)\right) = 0. 
\end{aligned}
\end{equation}
Thus we have
\begin{equation}
p_1^*=P_{1,w}(p_2^*). \label{optimal p1=W}
\end{equation}
To obtain $p_2^*$, we also need to discus three cases: i). $\lambda_4>0, \lambda_6=0$; ii). $\lambda_4=0, \lambda_6>0$ based on $a_1\geq\frac{E_{\max}B|h_2|^2}{(1+|h_1|^2p_1^*+|h_2|^2p_2^*)\ln(2)}$; iii). $\lambda_4>0, \lambda_6>0$. 
similarly, Case iii) can be included within case i) or ii).

\emph{\underline{Case 3}}: $\lambda_2=0, \lambda_5>0$, $\lambda_4>0$ and $\lambda_6=0$, to satisfy \eqref{lambdaPmax}, we have 
\begin{equation}
p_2^*=P_{\max}. %\label{optimal p2=pmax}
\end{equation}
In this case, to satisfy the feasible conditions \eqref{primalE1} and \eqref{primalE2}, we must have
\begin{equation}
p_1^*= P_{1,w}(p_2^*=P_{\max})\leq P_{\max}
\end{equation} 
and 
\begin{equation}
p_2^*=P_{\max}\leq P_{2,w}(p_1^*=P_{1,w}).
\end{equation} 
Therefore, we have the optimal solution as \eqref{Case3}.
%\begin{equation}
%\left\{
%\begin{aligned}
%&p_1^*=  P_{1,w}(p_2^*=P_{\max}) \quad  &\\
%&p_2^*=P_{\max} \quad   &
%\end{aligned}
%\right.
%\end{equation}
%based on the condition:
%\begin{equation}
%\left\{
%\begin{aligned}
%& P_{1,w}(p_2^*=P_{\max})\leq P_{\max} \quad  & \\
%& P_{\max}\leq P_{2,w}(p_1^*=P_{1,w}). &
%\end{aligned}
%\right.
%\end{equation}

\emph{\underline{Case 4}}: $\lambda_2=0, \lambda_5>0$, $\lambda_4=0$ and $\lambda_6>0$, to satisfy \eqref{lambdaE1} and \eqref{lambdaE2}, we have 
\begin{subequations}
\begin{align}\nonumber
\kappa_1a_1(f_1^{loc})^3+a_1p_1^*  - E_{\max}\left(b_1+B\log_2\left(1+|h_1|^2p_1^*+|h_2|^2p_2^*\right)\right) = 0\\
\kappa_2a_1(f_2^{loc})^3+a_1p_2^*  - E_{\max}\left(b_1+B\log_2\left(1+|h_1|^2p_1^*+|h_2|^2p_2^*\right)\right) = 0.\\\nonumber
\end{align}
\end{subequations}
In this case, we can derive the optimal solution Thus we have optimal solution \eqref{Case44}.
%This leads to 
%\begin{equation}\label{Case4}
%\left\{
%\begin{aligned}
%p_1^*=&\kappa_2(f_2^{loc})^3-\kappa_1(f_1^{loc})^3+p_2^*\\
%p_2^*=&-\frac{\mathcal{W}_0\left(-\frac{B_2\log(2)}{(|h_1|^2+|h_2|^2)}2^{\left(-\frac{B_2}{(|h_1|^2+|h_2|^2)}+A_2\right)}\right)}{B_2\log(2)}\\
%&-\frac{1+|h_1|^2\left(\kappa_2(f_2^{loc})^3-\kappa_1(f_1^{loc})^3\right)}{(|h_1|^2+|h_2|^2)}. 
%\end{aligned}
%\right.
%\end{equation}
%In this case, we need to satisfy primal feasible condition$ P_{1,w}(p_2^*)\leq P_{\max}$ and $ P_{2,w}(p_1^*)\leq P_{\max}$.
%Therefore, we can have the optimal solution \eqref{Case4} based on the condition:
%\begin{equation}
%\left\{
%\begin{aligned}
%& P_{1,w}(p_2^*)\leq P_{\max} \quad  & \\
%& P_{2,w}(p_1^*)\leq P_{\max}. &
%\end{aligned}
%\right.
%\end{equation}
Above all, the optimal solution of problem \eqref{Prob:p1,p2} can be concluded as four cases. 

\bibliographystyle{IEEEtran}
\bibliography{EEref}
\vspace{0.5em}
\end{document}